\documentclass{IEEEtran}

\usepackage{amsmath, amsthm, stmaryrd, syntax, paralist, graphicx}
\usepackage{tikz, graphicx, bussproofs, listings, balance}
\usepackage[caption=false]{subfig}
\usetikzlibrary{arrows}

\newcommand{\bc}[0]{B\&C}
\newcommand{\bct}[0]{BC($T$)}
\newcommand{\dpllt}[0]{DPLL($T$)}
\newcommand{\ie}[0]{\emph{i.e.}, }
\newcommand{\eg}[0]{\emph{e.g.}, }
\newcommand{\D}[0]{\ensuremath{\Delta}}
\newcommand{\kD}[0]{\ensuremath{k\text{-}\Delta}}

\newcommand{\inez}[0]{\textsf{Inez}}
\newcommand{\inezdb}[0]{\textsf{Inez\textit{DB}}}

\newcommand{\sem}[1]{\ensuremath{\llbracket #1 \rrbracket}}
\newcommand{\set}[1]{\ensuremath{\{#1\}}}
\newcommand{\setc}[2]{\ensuremath{\{#1\ |\ #2\}}}

\newcommand{\comment}[1]{}

\newtheorem{definition}{Definition}
\newtheorem{example}{Example}
\newtheorem{lemma}{Lemma}
\newtheorem{theorem}{Theorem}
\newtheorem{fact}{Fact}

\newcommand{\true}{\ensuremath{\mathsf{true}}}

\newcommand{\ED}{\ensuremath{\exists \D{}}}

\DeclareMathOperator{\oleft}{\mathsf{left}}
\DeclareMathOperator{\oright}{\mathsf{right}}
\DeclareMathOperator{\existsd}{\exists\hspace{-4pt}\exists}
\DeclareMathOperator{\omax}{\mathsf{max}}
\DeclareMathOperator{\omin}{\mathsf{min}}
\DeclareMathOperator{\ocount}{\mathsf{count}}
\DeclareMathOperator\lb{lb}
\DeclareMathOperator\ub{ub}
\DeclareMathOperator\lbp{lb^\prime}
\DeclareMathOperator\ubp{ub^\prime}
\DeclareMathOperator\rank{rank}
\DeclareMathOperator\match{match}

\newcommand\sfsm{\sf \small}

\newcommand{\sel}[3]{\langle \sigma\hspace{2pt}#1 : #2 : #3 \rangle}

\newcommand{\selll}[4]{
  \begin{aligned} 
    \langle \sigma\hspace{2pt}#1&:&& #2~\wedge& \\
    &&& #3 \\
    &:&& #4 &\rangle \\
  \end{aligned}
}

\newcommand{\sellll}[5]{
  \begin{aligned} 
    \langle \sigma\hspace{2pt}#1&:&& #2~\wedge& \\
    &&& #3~\wedge \\
    &&& #4 \\
    &:&& #5 &\rangle \\
  \end{aligned}
}

\newcommand{\parname}[1]{\noindent \textbf{#1}:}

\newcommand{\quotes}[0]{\ensuremath{\mathtt{quotes}}}
\newcommand{\portfolio}[0]{\ensuremath{\mathtt{portfolio}}}
\newcommand{\stocks}[0]{\ensuremath{\mathtt{stocks}}}

\makeatletter
\def\@citex[#1]#2{\leavevmode
\let\@citea\@empty
\@cite{\@for\@citeb:=#2\do
{\@citea\def\@citea{,\penalty\@m\ }%
\edef\@citeb{\expandafter\@firstofone\@citeb\@empty}%
\if@filesw\immediate\write\@auxout{\string\citation{\@citeb}}\fi
\@ifundefined{b@\@citeb}{\hbox{\reset@font\bfseries ?}%
\G@refundefinedtrue
\@latex@warning
{Citation `\@citeb' on page \thepage \space undefined}}%
{\@cite@ofmt{\csname b@\@citeb\endcsname}}}}{#1}}
\makeatother

\title{ILP Modulo Data}

\author{\IEEEauthorblockN{Panagiotis Manolios, Vasilis Papavasileiou,
    and Mirek Riedewald}\thanks{This research was supported in part by DARPA under AFRL
Cooperative Agreement No.~FA8750-10-2-0233 and by NSF grants
CCF-1117184 and CCF-1319580.
} \\
  \IEEEauthorblockA{Northeastern University\\
    \texttt{\{pete,vpap,mirek\}@ccs.neu.edu}}}

\begin{document}

\maketitle

\begin{abstract}
  The vast quantity of data generated and captured every day has led
  to a pressing need for tools and processes to organize, analyze and
  interrelate this data. Automated reasoning and optimization tools
  with inherent support for data could enable advancements in a
  variety of contexts, from data-backed decision making to
  data-intensive scientific research. To this end, we introduce a
  decidable logic aimed at database analysis. Our logic extends
  quantifier-free Linear Integer Arithmetic with operators from
  Relational Algebra, like selection and cross product. We provide a
  scalable decision procedure that is based on the \bct{} architecture
  for ILP Modulo Theories. Our decision procedure makes use of
  database techniques. We also experimentally evaluate our approach,
  and discuss potential applications.
\end{abstract}

\section{Introduction}

In 2010, enterprises and users stored more than 13 exabytes of
new data~\cite{bigdatawhitepaper}. Database Management Systems
(DBMS's) based on the Relational Model~\cite{codd} are a key component
in the computing infrastructure of virtually any organization.
With big data playing a determining role in business and science,
we are motivated to rethink data management and analysis.

\comment{SMT solvers like Z3~\cite{z3} enable interesting applications
  for various kinds of software by integrating arithmetic and
  software- or hardware- oriented theories (\eg lists and arrays).}

Database systems capable of symbolic computation could enable powerful
new methodologies for strategic planning, decision making, and
scientific research. We propose database systems that
\begin{inparaenum}[(a)]
\item store symbolic (in addition to concrete) data, and at the same
  time
\item allow queries of a symbolic nature, \eg with free variables.
\end{inparaenum} Such database systems can be dually thought of as
constraint solvers that reason in the presence of data. Symbolic data
allows us to encode partially specified or entirely speculative
information, \eg database entries that exist for the purpose of
what-if analysis. Symbolic queries enable deductive reasoning about
data.

Existing relational query languages (\eg SQL) only allow
concrete data and queries. Symbolic enhancements require a
formalism that combines constraints and relational queries. We address
this need by introducing the \D{}~logic. \D{}~extends quantifier-free
Linear Integer Arithmetic (QFLIA) with database tables and operators
from Relational Algebra, like selection ($\sigma$), union ($\cup$),
and cross product ($\times$). While \D{} is decidable, the logic in
its general form gives rise to hard satisfiability problems, primarily
because it allows universal quantification over cross products of big
tables. We study unrestricted~\D{} (for it is a natural umbrella
formalism), but also provide restrictions that enable an efficient
decision procedure. In other words, we identify a class of database
problems that are a realistic initial target for formal analysis.

We provide a scalable procedure based on the \bct{} architecture for
ILP Modulo Theories (IMT)~\cite{cav2013}. Our approach is dubbed
\emph{ILP Modulo Data}, because an ILP solver co-exists with a
procedure that establishes a correspondence between integer variables
and database tables. The latter contain a mix of concrete and symbolic
data. ILP Modulo Data allows us to use a powerful ILP solver based on
branch-and-cut (\bc{}) on the arithmetic side, while also utilizing
database techniques that allow us to scale to realistic datasets.

The compositional nature of ILP Modulo Data is well-suited for
potential applications. Organizations have access to vast amounts
of data, but at the same time rely heavily on Mathematical
Programming technology. We enhance Mathematical Programming tools
with the ability to directly access data, thus assisting
data-backed decision making. Such tools would also benefit
scientists in fields ranging from ornithology~\cite{mirekdddm} to
astronomy~\cite{skyserver}, by providing immediate feedback on
the consistency between models the scientists devise and datasets
of observations they collect. Our paper outlines potential
applications, while our experimental evaluation relies on
benchmarks that characterize them. We experimentally demonstrate
that our ILP Modulo Data framework provides better performance
than the approach of eagerly reducing \D{} to QFLIA.

\vspace{4pt}
\parname{Paper Structure} Section~\ref{sec:motivation} introduces our
reasoning paradigm through a motivating
example. Section~\ref{sec:logic} presents the \D{}~logic, while
Section~\ref{sec:ed} identifies a \D{}~fragment that yields scalable
procedures. Section~\ref{sec:bct} describes our decision procedure. We
experimentally evaluate our approach in
Section~\ref{sec:experiments}. We provide an overview of related work
in Section~\ref{sec:related}, and conclude with
Section~\ref{sec:conclusions}.

\section{Motivating Example}
\label{sec:motivation}

Our motivating example (formalized in Figure~\ref{fig:portfolio})
concerns the problem of optimally investing a given amount of
capital. This is an appropriate application for our techniques, because
\begin{inparaenum}[(a)]
\item investments are almost always data-driven as they take
  historical stock prices into account, and
\item financial institutions already rely on Mathematical Programming.
\end{inparaenum}

The problem involves investing in a \emph{portfolio} of $n$ publicly
traded stocks, with the goal of maximizing profit while following
guidelines that minimize risk. A database provides information on
these stocks, including stock prices from the New York Stock Exchange
(NYSE). We would like to pick the $n$ stocks that would have yielded
the highest profit over a period of time in the recent past, \eg over
the preceding year. This optimization problem is subject to
risk-mitigation constraints that require us to pick companies from a
variety of sectors. While investing in the exact solver-generated
portfolio (which relies only on past performance) is not necessarily a
good strategy, such a portfolio provides useful information for the
analysts who make the final investment decisions.

\begin{figure}[t]

  \center

  \subfloat [\stocks{}] {
    \label{fig:portfolio:stocks}
    \boxed{
      \begin{tabular}{l|l|l}
        Id & Cap & Sector \\
        \hline
        1 ({\sfsm EMC}) & \sfsm large & \sfsm tech \\
        2 ({\sfsm FII}) & \sfsm medium & \sfsm financials \\
        3 ({\sfsm AKR}) & \sfsm small & \sfsm retail \\
        \ldots & \ldots & \ldots \\
      \end{tabular}
    }
  }
  ~
  \subfloat [\quotes{}] {
    \label{fig:portfolio:quotes}
    \boxed{
      \begin{tabular}{l|l}
        Id & Diff \\
        \hline
        1 & 128 \\
        2 & 117 \\
        3 & 89 \\
        \ldots & \ldots \\
      \end{tabular}
    }
  }

  \subfloat [Constraints] {

    \small

    \(
    \begin{array}{ll}
      \textbf{maximize} & \\
      \hspace{4pt} \Sigma_{1 \leq i \leq n} a_i \cdot d_i \\
      \textbf{subject to} & \\
      \hspace{4pt}
      (x_i, c_i, s_i) \in \stocks{}, &\ 1 \leq i \leq n \\
      \hspace{4pt}
      (x_i, d_i) \in \quotes{}, &\ 1 \leq i \leq n \\
      \hspace{4pt}
      x_i \neq x_j, &\ 1 \leq i < j \leq n \\
      \hspace{4pt}
      \Sigma_{\setc{i}{1 \leq i \leq n, s_i = s}}
      a_i \leq \Sigma_{1 \leq i \leq n} a_i  / 3, & \text{ for every
        sector $s$} \\
      \hspace{4pt}
      \Sigma_{\setc{i}{1 \leq i \leq n, c_i = {\mbox{\sf \footnotesize
              small}}}}
      a_i \leq \Sigma_{1 \leq i \leq n} a_i  / 4
    \end{array}
    \)

  }
  
  \caption{Portfolio Management with ILP Modulo Data}
  \label{fig:portfolio}

\end{figure}

The data is given in tables \stocks{} and \quotes{}
(Figures~\ref{fig:portfolio:stocks}
and~\ref{fig:portfolio:quotes}). Each company in \stocks{} is
described by a unique ID (with the associated NYSE symbol
parenthesized), its capitalization ({\sfsm small}, {\sfsm medium}, or
{\sfsm large}), and its sector (\eg {\sfsm tech}, {\sfsm retail},
{\sfsm financials}, {\sfsm automotive}, {\sfsm energy}, {\sfsm
  emerging-markets}). While Figure~\ref{fig:portfolio} uses
human-readable names, we can encode these fields with bounded integer
quantities. Each entry in \quotes{} describes the observed movement of
a certain stock in a given timeframe, assuming that dividends were
reinvested. For example, the first row describes an increase of 28\%
in the price of {\sfsm EMC}. \quotes{} is an application-specific
abstraction, \ie the actual database contains past stock prices and
\quotes{} is a \emph{view} produced by comparing data for two time
periods.

The $i^{th}$ stock in the portfolio is characterized by a unique ID
$x_i$ that corresponds to entries in the dataset, \ie there exist
entries $(x_i, c_i, s_i) \in \stocks{}$ and $(x_i, d_i) \in
\quotes{}$. \comment{($c_i$, $s_i$, and $d_i$ are matching values for
  capitalization, sector, and price difference, respectively.)} To
minimize risk, we force the $n$ IDs $x_i$ to be distinct, and allow no
single sector to account for more than a third of the total capital.
Additionally, no more than a fourth of the capital goes to smallcap
companies. \comment{(The latter tend to be riskier investments.)} The
objective function maximizes the capital at the end of the period, and
thus the profit.

Note that if the amounts $a_i$ are variables, the objective function
is non-linear. The problem can be circumvented by providing integer
constants for $a_i$, \ie by specifying how the capital will be
partitioned. \comment{(We still have to determine which stock
  corresponds to each $a_i$.)} With constants for $a_i$, the non-table
constraints are essentially in QFLIA. (The summations for $i$ that
satisfy conditions like $s_i = s$ and $c_i = \mbox{\sfsm small}$ are
easy to encode as sums of if-then-else terms.) Conversely, the problem
is essentially satisfiability of an arithmetic instance, where certain
variables correspond to database contents. This is the kind of problem
that we propose new techniques for. We cannot use a standalone DBMS,
since DBMS's do not handle constraints and optimization. Neither are
existing solvers up to the task, since they do not provide ways of
managing data.

The constraints we have described are meant to be
representative. Clearly investors also have to consider other options,
including investing in index funds, bonds, debt securities and
derivative contracts. These financial instruments may have other
characteristics that need to be modeled. Our constraints are also
based on simplifying assumptions, \eg that we can invest an arbitrary
amount in any given stock at any time. It is not within the scope of
our paper to model investment problems comprehensively. What matters
is that these additional concerns also mix arithmetic with data, thus
reinforcing the need for data-aware solving.

\section{The Logic~\D{}}
\label{sec:logic}

\begin{figure}[h]

  \center

  \begin{minipage}{.8\columnwidth}
    \begin{grammar}

      <$F$> ::= $T_1 \leq T_2$ | $\existsd D$ | $\neg F$ | $F_1 \vee
      F_2$

      <$D$> ::= $\set{T^{+}}$ | $\sel{x}{F}{D}$ |
      $D_1 \times D_2$ | $D_1 \cup D_2$

      <$T$> ::= $(T_1, T_2)$ | $\oleft(T)$ | $\oright(T)$ | \\
      \vspace{2pt} \hspace*{2pt} $x$ | $K$ | $K \cdot T$ | $T_1 + T_2$ 

      <$K$> ::= $\ldots$ | $-2$ | $-1$ | $0$ | $1$ | $2$ | $\ldots$

    \end{grammar}
  \end{minipage}

  \caption{Grammar of~\D{}}
  \label{fig:syntax}

\end{figure}

This Section introduces the logic~\D{}. \D{} combines arithmetic with
queries over tabular data. \D{} thus encompasses database problems
like our motivating example of Section~\ref{sec:motivation}.

The grammar of~\D{} is given in Figure~\ref{fig:syntax}.  $K$, $T$,
$D$, and $F$ are the non-terminal symbols for integer constants,
terms, tables, and formulas, respectively. The first line of
productions for $T$ corresponds to pairs and their accessors; the
second line is for variable symbols ($x$) and integer expressions. A
table (non-terminal symbol $D$) is either an \emph{input table}, a
\emph{selection}, a \emph{cross-product}, or a \emph{union}. The
selection $\sel{x}{F}{D}$ is a table that consists of only those
entries in $D$ that satisfy $F$, \ie the variable $x$ ranges over the
table entries; $\sigma$ binds $x$ in $F$, but not in $D$. For formulas
(non-terminal symbol $F$), $\existsd D$ should be read as ``$D$ is not
empty''. \comment{$\existsd$ introduces a form of quantification, as
  we will see later.} All other constructs bear the obvious
meaning. We assume that all variables not bound by $\sigma$ are
integer. \comment{We call such variables \emph{free}.} We will freely
use derived operators, \eg conjunction and integer equality.

\D{} is typed. Each term is either of type \texttt{int} or of type $s
* t$, where $s$ and $t$ are types. $\oleft{}$ and $\oright{}$ are only
permissible when applied to a term of type $s * t$ for some type $s$
and some type $t$; if $x$ is of type $s * t$, then $\oleft(x)$ is of
type $s$ and $\oright(x)$ is of type $t$. The integer constants are of
type \texttt{int}.  The arithmetic operators ($+$, $\cdot$, and
$\leq$) only apply to terms of type \texttt{int}; $+$ and $\cdot$
produce integers. Each table has a \emph{schema}, which is the type of
its entries. (Schemas are the table-level counterpart of types.) An
input table is comprised of entries of the same type. If table $D_1$
has schema $s_1$ and table $D_2$ has schema $s_2$, then $D_1 \times
D_2$ has schema $s_1 * s_2$. For $\sel{x}{F}{D}$ to be properly typed,
$F$ should be a properly-typed formula under the assumption that the
type of $x$ is the schema of $D$; the schema of $\sel{x}{F}{D}$ is the
same as the schema of $D$. Union expects tables of the same schema and
preserves it.

Clearly, \D{} is at least as powerful as QFLIA. At the same time, \D{}
encompasses most features one would expect from a relational query
language. We have left out certain operators usually present in query
languages.  First, note that projection ($\pi$) would not provide
additional power, since it is possible to refer to any subset of the
columns, without producing an intermediate table that leaves out the
irrelevant ones. Also, the set difference $A \setminus B$ can be
encoded as $\sel{a}{\neg \existsd \sel{b}{a = b}{B}}{A}$, assuming
that the schema of $A$ and $B$ has exactly one column; otherwise, in
place of $a = b$ we would have a conjunction of equalities over all
columns. Additionally, \D{} can express many forms of aggregation,
including $\ocount$ (when compared to a constant), $\omin$, and
$\omax$.

\begin{example}
  The portfolio encoded by Figure~\ref{fig:portfolio} can be
  represented as the input table $$\portfolio{} = \{(1, (x_1, a_1)),
  \ldots, (n, (x_n, a_n))\}.$$ \portfolio{} contains symbolic data,
  something which is not allowed by DBMS's. The first column ensures
  that the $n$ entries are distinct, irrespective of the
  assignment. \portfolio{} is of schema $\mathtt{int} *
  (\mathtt{int} * \mathtt{int})$.  Consider the following constraint:
  $$
  \neg \existsd \boxed{ \selll{ x } { \oleft(\oleft(x)) \neq
      \oleft(\oright(x)) } { \oleft(\oright(\oleft(x))) =
      \oleft(\oright(\oright(x))) } { \boxed{\portfolio{}} \times
      \boxed{\portfolio{}} } }$$
  The constraint states that there are no entries $(i, (x_i, a_i))$
  and $(j, (x_j, a_j))$ in \portfolio{} such that $i \neq j$ and $x_i
  = x_j$, \ie \portfolio{} references $n$ distinct stocks (as was our
  intention in Figure~\ref{fig:portfolio}). The constraint essentially
  involves universal quantification over $\portfolio{} \times
  \portfolio{}$.
  \label{ex:distinct}
\end{example}

\subsection{Decidability}
\label{ssec:decidability}

\D{}~satisfiability can be reduced to QFLIA satisfiability. We explain
the reduction briefly. We represent a table expression $D$ of schema
$s$ as a set $\sem{D}$ consisting of pairs $r \oslash b$, where $r$ is
a term of type $s$ and $b$ is a QFLIA formula, with the intended
meaning that $r$ is present in the table iff $b$ is true. We use the
operator $\oslash$ to distinguish the auxiliary pairs used for the
reduction from the ones allowed by the syntax of \D{}. For a formula
$F$, $\sem{F}$ denotes the corresponding formula in QFLIA;
similarly for integer terms. $F[x / r]$ stands for substituting $x$
with $r$ in $F$, with appropriate care for occurrences of the
symbol $x$ bound by $\sigma$ inside $F$. We define $\sem{\cdot}$
for tables and formulas below as two mutually recursive functions.
\begin{equation}
  \begin{aligned}
    \sem{\set{r_1, \ldots, r_n}} = &\ 
    \{r_1 \oslash \true{}, \ldots, r_n \oslash \true{}\} \\
    \sem{\sel{x}{F}{D}} = &\ 
    \{r \oslash (b \wedge \sem{F[x / r]})\ |\
    r \oslash b \in \sem{D} \} \\
    \sem{D_1 \times D_2} = &\ 
    \{(r_1, r_2) \oslash (b_1 \wedge b_2)\ |\ \\
    &\hspace{16pt}r_1 \oslash b_1 \in
    \sem{D_1}, r_2 \oslash b_2 \in \sem{D_2} \} \\
    \sem{D_1 \cup D_2} = &\ \sem{D_1} \cup \sem{D_2}
  \end{aligned}
  \label{eq:sem:table}
\end{equation}
\vspace{6pt}
\begin{equation}
  \begin{aligned}
    \sem{T_1 \leq T_2} = & \sem{T_1} \leq \sem{T_2} \\
    \sem{\existsd D} = & \bigvee_{r \oslash b \in \sem{D}} b \\
    \sem{\neg F} = & \neg \sem{F} \\
    \sem{F_1 \vee F_2} = & \sem{F_1} \vee \sem{F_2}
  \end{aligned}
  \label{eq:sem:formula}
\end{equation}
For encoding \D{} integer terms as QFLIA terms (\eg $\sem{T_i}$ in
Equation~\ref{eq:sem:formula}), all that needs to be done is
elimination of pair constructors and accessors via the rules
$\oleft((x, y)) = x$ and $\oright((x, y)) = y$. The reduction suffices
to establish decidability of~\D{}.  The reduction also provides formal
semantics for~\D{} by specifying its meaning in terms of QFLIA.

\subsection{Complexity}

\begin{theorem}
  \label{thm:nexptime}
  The satisfiability problem for \D{} is in NEXPTIME.
\end{theorem}
\begin{proof}[Proof Sketch]
  The reduction to QFLIA (Equations~\ref{eq:sem:table}
  and~\ref{eq:sem:formula}) produces a formula exponentially larger
  than the input. Since QFLIA is in NP, the reduction provides a
  non-deterministic exponential time procedure
  for~\D{}-satisfiability.
\end{proof}

\begin{theorem}
  \label{thm:pspace}
  The satisfiability problem for \D{} is PSPACE-hard.
\end{theorem}
\begin{proof}[Proof Sketch]
  We reduce the (PSPACE-complete) QBF problem to \D{}~satisfiability
  in polynomial time. We deal with Boolean quantification by
  quantifying over the input table $\mathcal{B} = \{0, 1\}$. For
  example, the formula $\forall x \exists y (x \vee \neg y)$ becomes
  $\neg \existsd
  \boxed{
    \sel{x}{
      \neg \existsd
      \boxed{
        \sel{y}{x = 1 \vee y = 0}{\mathcal{B}}
      }
    }{
      \mathcal{B}
    }
  }.$
\end{proof}

Complexity analysis of~\D{} beyond Theorems~\ref{thm:nexptime}
and~\ref{thm:pspace} is not within the scope of this paper, and has
mostly theoretical significance. In practice, query size is orders of
magnitude smaller than data size. Conversely, it is meaningful to
study \emph{data complexity}~\cite{vardicomplexity}, \ie complexity
where only the amount of data varies. Instead of assuming a query of
constant size, we provide a stronger result by limiting the number of
tables that can participate in a cross product. (We also limit nested
quantifiers, because the latter can simulate cross products.)  We
define below the $\rank$ function that characterizes this number.
\begin{equation}
  \begin{aligned}
    \rank(\set{r_1, \ldots, r_n}) = &\ 1 \\
    \rank(\sel{x}{F}{D}) = &\ \rank(F) + \rank(D) \\
    \rank(D_1 \times D_2) = &\ \rank(D_1) + \rank(D_2) \\
    \rank(D_1 \cup D_2) = &\ \max(\rank(D_1), \rank(D_2))
  \end{aligned}
  \label{eq:rankd}
\end{equation}
\begin{equation}
  \begin{aligned}
    \rank(T_1 \leq T_2) = &\ 0 \\
    \rank(\existsd D) = &\ \rank(D) \\
    \rank(\neg F) = &\ \rank(F) \\
    \rank(F_1 \vee F_2) = &\ \max(\rank(F_1), \rank(F_2))
  \end{aligned}
  \label{eq:rankf}
\end{equation}

\begin{definition}[\kD{}]
  For any natural number $k$, \kD{} is the set of formulas $\setc{F}{F
    \in \D{}\ \text{and}\ \rank(F) \leq k}$.
\end{definition}

\begin{theorem}
  For any natural number $k$, \kD{} is NP-complete.
\end{theorem}
\begin{proof}[Proof Sketch]
  \kD{} is NP-hard , because any QFLIA formula can be reduced to a
  $0-\D{}$ formula in polynomial time ($0-\D{} \subseteq \kD{}$). We
  obtain membership in NP from the reduction defined by
  Equation~\ref{eq:sem:formula}, which produces polynomially-sized
  QFLIA formulas.
\end{proof}

\comment{\noindent Note that a very restricted fragment of~\D{}
  suffices to encode QBF. All the arithmetic needed is atomic formulas
  of the form $v = i$, where $v$ is a variable constrained in $\{0,
  1\}$ and $i$ is one of the constants $0$ and $1$; effectively, this
  is propositional reasoning. \comment{It is unlikely that
    satisfiability for unrestricted~\D{} is in PSPACE.}}

Given the class of formulas \kD{} for some $k$, the reduction produces
QFLIA formulas of size $O(n^{k+1})$, where $n$ is the input
size. While the reduction is polynomial (since $k$ is fixed), it may
not be practical even for $k = 2$, given that datasets of millions of
entries are common. Conversely, we propose restrictions that yield a
lazy solving architecture.

\section{The Existential Fragment of \D{}}
\label{sec:ed}

We proceed to study the \emph{existential fragment} of \D{}, which we
denote by \ED{}.

\begin{definition}[\ED{}]
  A \D{}~formula belongs to \ED{} if the $\existsd$ operator always
  appears below an even number of negations, \ie $\existsd$ only
  appears with positive polarity.
  \label{def:ed}
\end{definition}

The motivation for studying \ED{} is as follows. Universal
quantification pushes for an approach similar to quantifier
instantiation, \eg Example~\ref{ex:distinct} (which is not in~\ED{})
inherently requires instantiating a constraint for every element
in~$\portfolio{} \times \portfolio{}$. This can be done incrementally
by applying patterns that are standard in verification tools. In
contrast, we are not aware of techniques that would be a good match
for the kind of existential quantification that arises
in~\D{}. Therefore, the rest of this paper focuses on~\ED{}.

Formulas in~\ED{} can be transformed into formulas in a convenient
intermediate logic without cross products, selections, or unions. We
rephrase $\existsd$ in terms of a new membership operator. Each
formula of the form $\existsd D$ is viewed as $x \in D$, where $\in$
has the obvious semantics and $x$ is a properly shaped row comprised
of fresh integer variables. We will refer to rows like $x$ that serve
as witnesses for $\existsd$ as \emph{witness rows}. \comment{(Witness
  rows are analogous to Skolem constants.)} The next step is to
translate membership in arbitrary table expressions to membership in
input tables. $(x, y) \in D \times E$ becomes $x \in D \wedge y \in
E$, while $x \in D \cup E$ becomes $x \in D \vee x \in E$. Finally, $x
\in \sel{y}{F}{D}$ becomes $F[y / x] \wedge x \in D$. We eliminate all
cross products, selections, and unions by repeated application of the
above transformations.

\begin{example}
  The tables of Figures~\ref{fig:portfolio:stocks}
  and~\ref{fig:portfolio:quotes} can be easily encoded as \D{}~input
  tables of schemas $\mathtt{int} * (\mathtt{int} * \mathtt{int})$ and
  $\mathtt{int} * \mathtt{int}$. Let {\sfsm small} capitalization be
  represented by the constant $0$. Consider the following constraint:
  $$
  \existsd
  \boxed{
    \sellll{
      x
    } {
      \oleft(\oleft(x)) = \oleft(\oright(x))
    } {
      \oleft(\oright(\oleft(x))) = 0
    } {
      \oright(\oright(x)) \geq 150
    } {
      \boxed{\stocks{}} \times \boxed{\quotes{}}
    }
  }
  $$
  The constraint asserts the existence of some tuple $((x_1, (x_2,
  x_3)), (x_4, x_5)) \in \stocks{} \times \quotes{}$ that satisfies
  $\Phi = [x_1 = x_4 \wedge x_2 = 0 \wedge x_5 \geq 150]$. (We have
  eliminated the accessors $\oleft$ and $\oright$.)  This is
  equivalent to asserting that $(x_1, (x_2, x_3)) \in \stocks{} \wedge
  (x_4, x_5) \in \quotes{} \wedge \Phi$.
\end{example}

The procedure we outlined produces a \emph{decomposed} formula
consisting of a QFLIA part and \emph{membership constraints}. We
proceed to define these notions formally.

\begin{definition}[(Conditional) Membership Constraint]
  A membership constraint is a constraint of the form
  \begin{equation}
    (x_1,
    \ldots, x_k) \in \{(y_{1,1}, \ldots, y_{1,k}), \ldots, (y_{l,1},
    \ldots, y_{l,k}) \}
    \label{eq:mem}
  \end{equation}
  for positive integers $k$ and $l$ and variable symbols $x_i$,
  $y_{j,i}$. A constraint of the form $b = 1 \Rightarrow m$, where $b$
  is a variable symbol and $m$ is a membership constraint, is called a
  \emph{conditional} membership constraint.
  \label{def:mem}
\end{definition}
A membership constraint may hold conditionally, either because it
arises from an $\existsd$-atom that appears under propositional
structure (and therefore holds conditionally), or because of a
disjunction introduced by the union operator. We use conditions of the
form $b = 1$ because ILP necessitates $[0, 1]$-bounded integer
variables in place of Boolean variables.  Implication in the opposite
direction is never needed, since $\existsd$ always appears with
positive polarity (as per Definition~\ref{def:ed}).

Membership constraints do not contain arbitrary arithmetic
expressions, but only variable symbols. ``Variable
abstraction''~\cite{combiningdp} eliminates richer expressions. While
variable abstraction allows for compositional reasoning and helps with
theoretical analysis, a limited fragment of arithmetic in membership
constraints yields more efficient implementation. Part of our
discussion will involve tables that contain integer constants and
terms of the form $v + c$, where $v$ is a variable symbol and $c$ is
an integer constant. (Everything we present is easy to generalize for
such terms.) For convenience, we flatten out rows constructed using
the pair constructor of Figure~\ref{fig:syntax}, and instead deal with
$k$-tuples of integers. This is only a matter of presentation and has
no impact on the algorithms.

\begin{definition}
  A \emph{decomposed} formula is a conjunction $F \wedge M$, where
  \begin{inparaenum}[(a)]
  \item $F$ is a QFLIA formula and
  \item $M$ is a conjunction of possibly conditional membership
    constraints.
  \end{inparaenum}
  \label{def:decomposed}
\end{definition}

\begin{theorem}
  \ED{}~satisfiability is NP-complete.
  \label{thm:npcomplete}
\end{theorem}

\begin{proof}
  \ED{}~satisfiability is NP-hard, because \ED{} is at least as
  powerful as QFLIA. \ED{}~satisfiability is in NP, because we can
  reduce \ED{} to QFLIA in polynomial time. The reduction first
  produces a formula in decomposed form
  (Definition~\ref{def:decomposed}). Equation~\ref{eq:mem} is
  equivalent to $\bigvee_{j = 1, \ldots, l} \bigwedge_{i = 1, \ldots,
    k} x_i = y_{j,i}$; therefore, the membership operator can be
  eliminated. The result is a formula in QFLIA.
\end{proof}

The polynomial size of the reduction relies on the fact that~\D{} does
not allow tables to be named and referenced from multiple places, \ie
table expressions are not DAG-shaped. Despite the polynomial
reduction, a lazy scheme remains relevant. The reason is that QFLIA
solvers are not meant for long disjunctions that essentially encode
database tables.

\section{\bct{} for \D{}}
\label{sec:bct}

The decomposed form of Definition~\ref{def:decomposed} is particularly
suited for a scheme that combines separate procedures for QFLIA and
table membership. Given that the QFLIA part can be encoded as a
conjunction of integer linear constraints~\cite{cav2013}, it becomes
possible to solve instances in decomposed form (and by extension
\ED{}~instances) by instantiating the \bct{} framework for
IMT~\cite{cav2013}. An ILP solver deals with the QFLIA constraints,
and exchanges information with a procedure that checks membership in
finite sets. Since database queries typically have simple
propositional structure, we do not expect encoding the latter with
linear constraints to be a bottleneck.

The membership procedure is confronted with a conjunction of
membership constraints (Definition~\ref{def:mem}). Dealing with
conditional constraints is essentially a matter of Boolean search.
The membership procedure needs to understand equality atoms, equality
being a primitive. (Our setting is standard first-order logic with
equality.)  In particular, the procedure keeps track of truth
assignments to the equalities in:
\begin{equation}
  \{x_i = y_{j,i}\ |\ j \in [1, l], i \in [1, k]\}
  \label{eq:literals}
\end{equation}
The symbols $x_i$ and $y_{j,i}$ have the same meaning as in
Definition~\ref{def:mem}. In the presence of multiple membership
constraints, the union of sets, like  in
Equation~\ref{eq:literals}, is relevant. Given that membership
constraints can be checked in isolation, our discussion proceeds with
a single constraint. The variables $x_i$ and $y_{j,i}$ also appear in
linear constraints. It simplifies our design to assume that all of
them appear in ILP, even if they are unconstrained there. The \bct{}
framework provides a mechanism (``difference
constraints''~\cite{cav2013}) for notifying background procedures
about atoms like the ones in Equation~\ref{eq:literals}. Given truth
values for these atoms, we check that a membership constraint is
satisfied by simply traversing the table and looking for a tuple that
is column-wise equal to the witness row. The constraint is violated if
for every $j \in [1, l]$, there exists some $i \in [1, k]$ such that
$x_i \neq y_{j,i}$, \ie there is no candidate tuple.

The arithmetic and membership parts share variables. It is vital that
we systematically explore the space of (dis)equalities between these
variables. This exchange of information resembles the
non-deterministic Nelson-Oppen scheme (NO) for combining decision
procedures~\cite{no79}. We demonstrate that NO can accommodate
membership constraints.

\begin{definition}[Arrangement]
  Let $E$ be an equivalence relation over a set of variables $V$. The
  set
  $$
  \alpha(V, E) = \setc{x = y}{x E y}\ \cup\ \setc{x \neq y}{x, y \in V
    \text{ and not } x E y}
  $$
  is the \emph{arrangement} of $V$ induced by $E$.
\end{definition}

\begin{definition}[Stably-Infinite Theory]
  A $\Sigma$-theory $T$ is called stably-infinite if for every
  $T$-satisfiable quantifier-free $\Sigma$-formula $F$ there exists an
  interpretation satisfying $F \wedge T$ whose domain is infinite.
\end{definition}

\begin{fact}[Nelson-Oppen for Stably-Infinite
  Theories~\cite{no79,combiningdp}]
  Let $T_i$ be a stably-infinite $\Sigma_i$-theory, for $i = 1, 2$,
  and let $\Sigma_1 \cap \Sigma_2 = \emptyset$. Also, let $\Gamma_i$
  be a conjunction of $\Sigma_i$-literals. $\Gamma_1 \cup \Gamma_2$ is
  $(T_1 \cup T_2)$-satisfiable iff there exists an equivalence
  relation $E$ of the variables $V$ shared by $\Gamma_1$ and
  $\Gamma_2$ such that $\Gamma_i \cup \alpha(V, E)$ is
  $T_i$-satisfiable, for $i = 1, 2$.
  \label{fact:no}
\end{fact}

\begin{lemma}[Nelson-Oppen with Propositional Structure]
  Let $T_i$ be a stably-infinite $\Sigma_i$-theory, for $i = 1, 2$,
  and let $\Sigma_1 \cap \Sigma_2 = \emptyset$. Also, let $\Phi_i$
  be a quantifier-free $\Sigma_i$-formula. $\Phi_1 \wedge \Phi_2$ is
  $(T_1 \cup T_2)$-satisfiable iff there exists an equivalence
  relation $E$ of the variables $V$ shared by $\Phi_1$ and $\Phi_2$
  such that $\set{\Phi_i} \cup \alpha(V, E)$ is $T_i$-satisfiable,
  for $i = 1, 2$.
  \label{lemma:noprop}
\end{lemma}

\begin{proof}
  \hfill
  \begin{itemize}
  \item If $\Phi_1 \wedge \Phi_2$ is $(T_1 \cup T_2)$-satisfiable,
    there exists a first-order model $L$ that satisfies it. The way
    $L$ interprets the variables in $V$ gives rise to an equivalence
    relation $E$ over $V$ such that $L$ satisfies $T_i \cup
    \set{\Phi_i} \cup \alpha(V, E)$, $i = 1, 2$.
  \item If there exists an equivalence relation $E$ over $V$ such that
    $\set{\Phi_i} \cup \alpha(V, E)$ is $T_i$-satisfiable, then there
    exists a model $L_i$ that $T_i$-satisfies $\Phi_i$, $i = 1,
    2$. Let
    \begin{align*}
      \Gamma_i =~&~\{t\ |\ t \text{ is
        an atom in } \Phi_i, L_i \models t\}~~\cup &\\
      &~\{\neg t\ |\ t \text{ is
        an atom in } \Phi_i, L_i \models \neg t\},&i = 1, 2.
    \end{align*}
    $\Gamma_i \cup \alpha(V, E)$ is $T_i$-satisfiable, $i = 1, 2$. By
    Fact~\ref{fact:no}, $\Gamma_1 \cup \Gamma_2$ is $(T_1 \cup
    T_2)$-satisfiable. But $\Gamma_1 \cup \Gamma_2 \models \Phi_1
    \wedge \Phi_2$. Therefore, $\Phi_1 \wedge \Phi_2$ is $(T_1 \cup
    T_2)$-satisfiable.
  \end{itemize}
\end{proof}

\begin{lemma}[Nelson-Oppen for Membership Constraints]
  \hfill Let $T$ be a stably-infinite $\Sigma$-theory. Also, let
  $\Gamma$ be a conjunction of $\Sigma$-literals, and $M$ be a
  conjunction of possibly negated membership constraints.  $\Gamma
  \cup M$ is $T$-satisfiable iff there exists an equivalence relation
  $E$ of the variables $V$ shared by $\Gamma$ and $M$ such that
  $\Gamma \cup \alpha(V, E)$ is $T$-satisfiable and $M \cup \alpha(V,
  E)$ is satisfiable.
  \label{lemma:nomem}
\end{lemma}

\begin{proof}
  Membership constraints can be viewed as disjunctions of conjunctions
  (Proof of Theorem~\ref{thm:npcomplete}) in which no function,
  predicate, and constant symbols appear, \ie in the empty
  signature. The theory pertaining the membership constraints is the
  empty theory ($\emptyset$), since no axioms are needed. $\emptyset$
  is trivially stably-infinite. Our proof obligation follows by
  applying Lemma~\ref{lemma:noprop} with $T_1 = T$ and $T_2 =
  \emptyset$.
\end{proof}

Note that Lemma~\ref{lemma:nomem} allows negated membership
constraints. While the latter do not pose algorithmic difficulties,
our discussion is limited to the positive occurrences needed for
\ED{}. The statement of Lemma~\ref{lemma:nomem} is structurally
similar to that of Fact~\ref{fact:no}, with membership constraints
replacing the constraints of some participating stably-infinite
theory. It follows that a membership procedure can participate in NO
as a black box, much like a theory solver, even though we have not
formalized membership constraints by means of a theory. We can thus
combine a form of set reasoning with any stably-infinite theory.

\bct{} guarantees completeness for the combination of ILP with a
stably-infinite theory~\cite{cav2013} by ensuring that the branching
strategy explores all possible arrangements. We established that
membership can be used much like a stably-infinite theory. All that is
needed for completeness is a membership procedure capable of checking
consistency of its constraints conjoined with a given arrangement
(that contains all literals of Equation~\ref{eq:literals}). As we have
seen, this operation is simple and involves no arithmetic. In pursuit
of efficiency, we proceed to describe branching and propagation
techniques based on table contents. Meaningful branching and
propagation involve the integer bounds of variables, \ie necessitate
limited arithmetic reasoning on the membership side.

\subsection{Propagation}
\label{ssec:prop}

\bc{}-based ILP solvers keep track of variable lower and upper bounds,
and heavily rely on bounds propagation algorithms. We describe how to
enhance such propagation to exploit the structure of membership
constraints.

We denote by $\lb(v)$ an $\ub(v)$ the current lower and upper bounds
on variable~$v$. $\lb(v)$ (respectively $\ub(v)$) is either an integer
constant, or $-\infty$ (resp. $+\infty$) if no bound is known. We use
the notation $\lbp(v)$ and $\ubp(v)$ for bounds on $v$ that the
membership procedure infers. \comment{(The prime symbol indicates
  ``next state.'')} We proceed with a membership constraint as per
Definition~\ref{def:mem}.  Let $x = (x_1, \ldots, x_k)$; similarly, we
denote by $y_j$ the tuple $(y_{j,1}, \ldots, y_{j,k})$. Let $\match(x,
y_j)$ be true if and only if for all $i \in [1, k]$, the sets
$[\lb(x_i), \ub(x_i)]$ and $[\lb(y_{j,i}), \ub(y_{j,i})]$ intersect.

\vspace*{-1em}
\footnotesize
\begin{align}
  \lbp(x_i) = \max(\lb(x_i),&\min \{ \lb(y_{j,i})\ |\ j \in [1, l],
  \match(x, y_j) \})
  \label{eq:proplb} \\
  \ubp(x_i) = \min(\ub(x_i),&\max \{ \ub(y_{j,i})\ |\ j \in [1, l],
  \match(x, y_j) \})
  \label{eq:propub}
\end{align}
\vspace*{-16pt}
\normalsize

\noindent We over-approximate the values of the variables $x_i$ by
considering all candidate entries (inner $\min$ and $\max$). The outer
$\max$ and $\min$ guarantee that we do not weaken bounds. If there
exists exactly one value $j$ such that $\match(x, y_j)$, it is sound
to deduce the equalities $x_i = y_{j,i}$, for all $i \in [1, k]$. If
there is no candidate entry, inconsistency is reported.

\begin{example}[Interleaved Propagation]
  Consider the decomposed formula $x = y \wedge (x, y) \in \{(1, 2),
  (2, 4), (3, 6), (4, 8)\}$. The formula corresponds to a query over
  concrete tuples that any DBMS can evaluate in linear time. It is
  thus vital that our techniques yield acceptable
  performance. Equations~\ref{eq:proplb} and~\ref{eq:propub} bound $x$
  to $[\min \{1, 2, 3, 4\}, \max \{1, 2, 3, 4\}] = [1, 4]$ and y to
  $[\min \{2, 4, 6, 8\}, \max \{2, 4, 6, 8\}] = [2, 8]$.  Given the
  equality $x = y$, ILP propagation deduces that $x, y \in [2, 4]$,
  since $[2, 4]$ is the intersection of permissible ranges for $x$ and
  $y$. The membership procedure detects that $\match$ now only holds
  for $(2, 4)$, and fixes $x$ to $2$ and $y$ to $4$. The ILP solver in
  turn deduces unsatisfiability, since $x = y$ is violated. No
  branching was needed. Encoding the formula in QFLIA would hide its
  structure, leading to search. The example generalizes to other
  lengths and bounded symbolic data.
  \label{ex:interleavedprop}
\end{example}

\subsection{Branching and Arrangement Search}
\label{ssec:branching}

It follows from Lemma~\ref{lemma:nomem} that a branching strategy
which exhaustively explores all possible arrangements of the shared
variables guarantees completeness. To achieve better performance, we
have to branch with the tabular structure of databases in mind,
without overlooking symbolic data.

\begin{figure}

  \center

  \begin{tikzpicture}[auto, font=\scriptsize]

    \node[draw = black, dotted, label = right:$(0)$] (p0) at (0, 0)
    {
      \(
      (x_1, x_2) \in
      \boxed{
        \begin{aligned}
          \{&(1, 2), \\
          &(2, 3), \\
          &(3, 2), \\
          &(y_1, y_2) \}
        \end{aligned}
      }
      \)
    };

    \node[draw = black, dotted, label = right:$(1)$] (p1)
    at (-65pt, -58pt)
    {
      \(
      (x_1, x_2) \in
      \boxed{
        \begin{aligned}
          \{&(1, 2), \\
          &(y_1, y_2) \}
        \end{aligned}
      }
      \)
    };

    \node[draw = black, dotted, label = right:$(2)$] (p2)
    at (55pt, -58pt)
    {
      \(
      (x_1, x_2) \in
      \boxed{
        \begin{aligned}
          \{&(2, 3), \\
          &(3, 2), \\
          &(y_1, y_2) \}
        \end{aligned}
      }
      \)
    };

    \draw[->] (p0) -- (p1) node[midway, left, xshift = -8pt, yshift =
    2pt] {$x_1 < 2$};

    \draw[->] (p0) -- (p2) node[midway, right, xshift = 10pt, yshift =
    2pt] {$x_1 \geq 2$};

    \node[draw = black, dotted, label = right:$(3)$] (p3)
    at (-65pt, -118pt)
    {
      \(
      (x_1, x_2) \in
      \boxed{
        \begin{aligned}
          \{&(2, 3) \\
          &(3, 2)\} \\
        \end{aligned}
      }
      \)
    };

    \node[draw = black, dotted, label = right:$(4)$] (p4)
    at (55pt, -118pt)
    {
      \(
      (x_1, x_2) \in
      \boxed{
        \begin{aligned}
          \{&(2, 3) \\
          &(3, 2), \\
          &(y_1, y_2) \}
        \end{aligned}
      }
      \)
    };

    \draw[->] (p2) -- (p3) node[midway, left, xshift = -8pt]
    {$x_1 \neq y_1$};
    \draw[->] (p2) -- (p4) node[midway, right, xshift = 8pt]
    {$x_1 = y_1$};

  \end{tikzpicture}


  \caption{Data-Driven Branching}
  \label{fig:branching}

\end{figure}
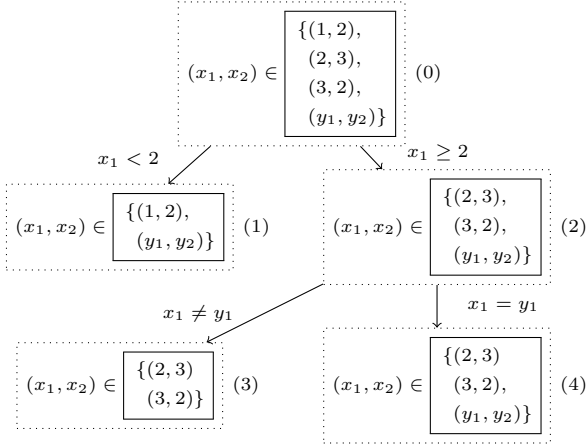

Figure~\ref{fig:branching} provides an example. The root node (Node
$0$) describes a single membership constraint, which we assume to be
part of a larger decomposed formula. We maintain integer constants in
the table, instead of performing variable abstraction which would
introduce auxiliary variables for them. According to
Equation~\ref{eq:literals}, the membership procedure needs truth
assignments for the equalities in $\{x_1 = 1, x_1 = 2, x_1 = 3, x_1 =
y_1, x_2 = 2, x_2 = 3, x_2 = y_2\}.$ It would not be wise for the
search strategy to overlook that this set originates from a table
containing numbers, and treat the set members as if they were atomic
propositions unrelated to each other.

In our example, branching on the condition $x_1 < 2$ produces two
subproblems. Node 1 shows only the tuples that still apply under the
condition $x_1 < 2$, \ie the ones that still satisfy the predicate
$\match$; similarly for Node 2. $x_1 < 2$ is a choice informed by the
tabular structure. Since $2$ as the value of the first column is close
to the ``middle'' of the table, branching on $x_1 < 2$ rules out
approximately half of the candidates. $(y_1, y_2)$ is present in both
subproblems (Nodes 1 and 2). Branching based on constant bounds is
therefore not enough, for we will possibly have to deal with symbolic
tuples. Figure~\ref{fig:branching} demonstrates further branching on
$x_1 = y_1$ to determine whether $(y_1, y_2)$ is a suitable witness
for the membership constraint.

The example demonstrates the dual nature of the search strategy
needed. The problem naturally pushes towards branch-and-bound (which
is a restriction of \bc{}), \eg branching on $x_1 < 2$ is
meaningful. It remains necessary to also branch on equalities between
shared variables (\eg $x_1 = y_1$), just like in any practical
implementation of NO.  (To be precise, in ILP we would have two
separate nodes for $x_1 < y_1$ and $x_1 > y_1$ in place of $x_1 \neq
y_1$.)  Implementing NO with \bc{} enables both kinds of branching.

Branching is organically tied to propagation. Initially (Node $0$),
assuming no previously known bounds for $x_1$, the table contents only
allow us to bound $x_1$ to the range $[\min(\lb(y_1), 1),
\max(\ub(y_1), 3)]$; if $y_1$ is unbounded, $x_1$ remains
unbounded. The decisions $x_1 \geq 2$ and $x_1 \neq y_1$ (\ie Node
$3$) tighten $x_1$ to $[2, 3]$. We also obtain the range $[2, 3]$ for
$x_2$, \ie branching on some column potentially leads to propagation
across other columns.

\subsection{Discussion}

The analysis of this Section indicates that \D{}~formulas can be
decomposed in such a way that a procedure for table lookup assumes
part of the workload. \bct{} is particularly suited for implementing
such a combination. \bct{} can easily accommodate data-aware
propagation (Section~\ref{ssec:prop}) and branching
(Section~\ref{ssec:branching}). Our techniques would be harder to
implement within a \dpllt{}-style solver~\cite{dpllt}, given that the
toplevel search of \dpllt{} is over the Booleans (and not the
integers). A \dpllt{}-based implementation of our techniques would
essentially require integrating branch-and-bound in \dpllt{}, which
is beyond the scope of our work.

The table lookup procedure can be thought of as a small database
engine within the solver. The employed database engine can be an
actual DBMS, storing the concrete part of tables and possibly bounds
on symbolic fields. A DBMS would provide multiple opportunities for
improvements.  Equations~\ref{eq:proplb} and~\ref{eq:propub}
essentially describe database aggregation, and thus provide a starting
point for the kinds of queries that apply. DBMS queries can be over
multiple tables at a time, and can involve conditions other than
bounds. As a matter of fact, the $\match$ predicate of
Equations~\ref{eq:proplb} and~\ref{eq:propub} can be strengthened with
any condition on the data that follows from the formula (\eg $x = y$
in Example~\ref{ex:interleavedprop}), thus computing tighter
bounds. Different kinds of database optimizations apply, \eg
materializing queries for better incremental behavior and smarter
indexing based on user input.

\ED{} (and its decomposed form) formally characterizes a relevant
class of problems that can be solved by a compositional scheme which
employs a database engine. Our scheme may actually apply to a superset
of~\ED{}.

\section{Applications and Experiments}
\label{sec:experiments}

We have implemented support for databases on top of the \inez{}
constraint
solver.\footnote{\texttt{https://github.com/vasilisp/inez}} \inez{}
is our implementation of the \bct{} architecture for IMT on top of the
SCIP (M)ILP solver~\cite{scip}. We refer to the version of \inez{}
that provides database extensions as \inezdb{}.  \inezdb{} supports
existential database constraints by means of the \bct{}-based
combination described in Section~\ref{sec:bct}, but also universal
quantification by eager instantiation. \inezdb{} (like \inez{})
additionally supports objective functions.

We have produced a collection of \inezdb{} input files that have the
structure we expect in applications. Our benchmark suite is publicly
available and can be used as a starting point towards a richer
benchmark suite of problems that involve data and
constraints.\footnote{\texttt{http://www.ccs.neu.edu/home/vpap/fmcad-2014.html}}
We provide a brief overview of the application areas that inspire our
benchmarks.

\vspace*{-10pt}
\subsection{How-To Analysis}

Research in the general direction of reverse data
management~\cite{revdata} proposes ways of obtaining the desired
results out of a database query. \comment{This can be achieved either
  by modifying the data~\cite{tiresias} or by changing the query
  itself~\cite{conquer}.} We outline this class of problems through an
example, which gives rise to some of our benchmarks.

\begin{example}[\texttt{emp\_join.ml}]
  The management of a company is surprised to find out that (according
  to the corporate database) there is no employee younger than 30
  whose yearly income exceeds \$60000. \emph{Why not} is not obvious,
  since income is a complicated function of multiple quantities
  including a base salary, benefits based on age, employee level
  (junior, middle, or senior), and bonuses.

  The management consults the database administrator on \emph{how
    to}~\cite{tiresias} ameliorate the seeming injustice. Together,
  they explore bonuses that would allow young employees to exceed the
  \$60000 limit. This amounts to synthesizing tuples for the table of
  bonuses. An alternative is to adjust various parameters in the
  income computation, \ie to modify the query instead of the
  data~\cite{conquer}. This can be done by replacing constants with
  variables, and letting the solver come up with suitable values.
\end{example}

\vspace*{-10pt}
\subsection{Test-Case Generation}

Test case generation is relevant for databases~\cite{qex}. A family of
benchmarks in our collection demonstrate test data generation by
concretizing tables initially containing symbolic data.

\begin{example}[\texttt{emp\_keys.ml}]
  The problem involves two tables, named $\mathtt{incomes}$ and
  $\mathtt{employees}$. $\mathtt{incomes}$ has an ID column
  constrained to reference existing entries in $\mathtt{employees}$,
  \ie there is a foreign key constraint. $\mathtt{incomes}$ contains
  thousands of tuples with symbolic IDs. A satisfying assignment
  corresponds to a generated database that meets the foreign key
  constraint, thus serving as meaningful test input.
\end{example}

\vspace*{-10pt}
\subsection{Scientific Applications}

Studying big datasets is a key aspect of scientific research in fields
ranging from ornithology~\cite{mirekdddm} to
astronomy~\cite{skyserver}. To demonstrate the applicability of our
techniques, we provide benchmarks inspired by queries that
ornithologists perform.

\begin{example}[\texttt{birds\_box.ml}]
  An ornithologist wants to see a rare species in person, but has not
  decided on a good location.  She has access to a database of
  observations. Each observation describes a bird and the geographic
  coordinates where it was seen. An area can be described as a
  symbolic rectangle $B = [\mathtt{longitude}_{\mathtt{min}},
  \mathtt{longitude}_{\mathtt{max}}] \times
  [\mathtt{latitude}_{\mathtt{min}},
  \mathtt{latitude}_{\mathtt{max}}]$. Our techniques allow the
  ornithologist to simply ask for $n$ observations of the species of
  interest that lie in $B$. The query effectively concretizes $B$.
\end{example}

\vspace*{-10pt}
\subsection{Portfolio Management}

We experimented with the portfolio optimization example of
Section~\ref{sec:motivation}. Our exact instance
(\texttt{portfolio.ml}) encodes a more complex variant of the
formalization in Section~\ref{sec:motivation}. An additional table
contains stock dividends; dividends are taken into account in the
objective function. We tried a range of parameters with a timeout of
one hour, and obtained a range of solutions. Notably, picking an
optimal portfolio of 5 out of 50 stocks took 161 seconds; 5 out of
4000 stocks took 1510 seconds; and 6 out of 2000 stocks took 1172
seconds. Such table sizes are realistic, given that NYSE lists
approximately 2800 companies.

\vspace*{-10pt}
\subsection{Overview of Results}

\begin{figure}[t]

  \center
  
  \subfloat [\inezdb{} versus \inez{}] {
    \includegraphics [trim= 0 10pt 0 10pt]
    {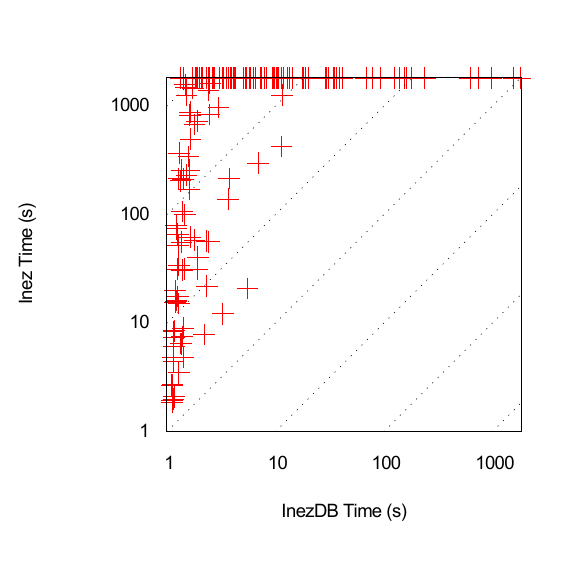}
    \label{fig:exp:vsinez}
  }

  \subfloat [\inezdb{} versus Z3] {
    \includegraphics [trim=0 10pt 0 10pt]
    {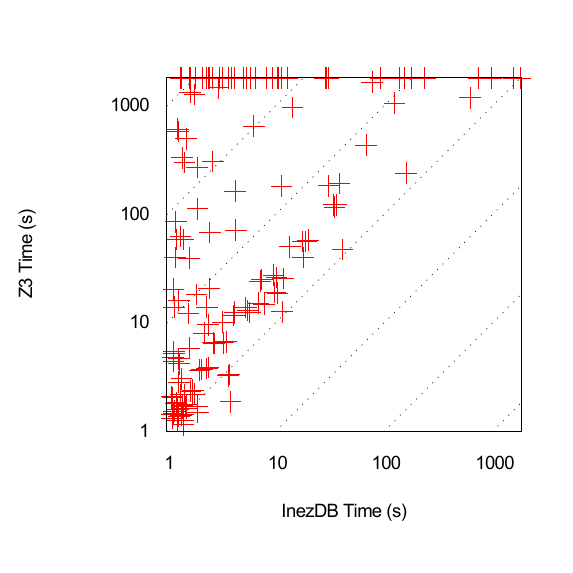}
    \label{fig:exp:vsz3}
  }

  \caption{Experiments: \inezdb{} versus the eager approach}
  \label{fig:exp}

\end{figure}

We compare \inezdb{} against an \inez{} frontend that solves \D{}
formulas by eagerly translating them to QFLIA via the encoding of
Theorem~\ref{thm:npcomplete}. \inez{} in turn solves QFLIA formulas by
reducing them to constraints that SCIP understands.  (These
constraints are not strictly ILP, since we utilize specialized
constraint handlers~\cite{scip}.)  We refer to this configuration
simply as \inez{}, since the only addition to \inez{} is a new
frontend. We also produce SMT-LIB versions of our QFLIA formulas, and
run them against the latest available version of Z3 (4.3.1).

We provide 8 benchmark generators that allow different modes of
operation (\eg some of them are able to produce both satisfiable and
unsatisfiable benchmarks), and are able to output benchmarks with
different table sizes. Our input table sizes range from 60 tuples to
640000 tuples. In total, our parameters give rise to 166
benchmarks. We run all three solvers with a timeout of 1800 seconds
and a memory limit of 12GB on a machine that provides 2 Intel Xeon
X5677 CPUs of 4 cores each and 96GB of RAM. Figure~\ref{fig:exp}
visualizes our experiments. \inez{} solves 25 satisfiable and 47
unsatisfiable benchmarks. \inezdb{} solves 74 satisfiable and 81
unsatisfiable benchmarks. Finally, Z3 solves 57 satisfiable and 58
unsatisfiable benchmarks. Among the failures for \inez{} (resp. Z3),
37 (resp. 27) are due to the memory limit. \inezdb{} runs out of
memory only once. If we turn off the memory limits, the total numbers
of failures don't change much.

\comment{\inez{} times out for 57 of the benchmarks, and runs out of
  memory for 37 of them.  \inezdb{} times out for 10 of the benchmarks
  and runs out of memory once. Finally, Z3 times out for 24 of the
  benchmarks and runs out of memory 27 times.}

Figure~\ref{fig:exp:vsinez} indicates that \inezdb{} outperforms
\inez{} by a significant margin. This margin can be attributed to two
factors. First, \inezdb{} exploits the structure of database problems
(\eg for branching and propagation), while \inez{} has no knowledge of
this structure. Second, our reduction to QFLIA (in the case of
\inez{}) produces patterns that SCIP is not optimized for, since the
latter is designed for MILP and not for QFLIA.

Figure~\ref{fig:exp:vsz3} compares \inez{} against a leading solver
for QFLIA (Z3), and thus characterizes the tool's performance in
absolute terms. There is a cluster of 40 benchmarks for which
\inezdb{} is 2-8 times faster than Z3. (Note that the scale is
logarithmic.) \inezdb{} is at least 8 times faster for 31 of the
benchmarks that both tools solve, and solves many benchmarks for which
Z3 times out. All failures for \inezdb{} are failures for Z3. Z3
outperforms \inezdb{} for only 7 out of the 166 benchmarks, none of
which take \inezdb{} more than 4 seconds to solve. \comment{In light
  of these experiments, our techniques of Section~\ref{sec:bct} do not
  just allow \inezdb{} to outperform \inez{}, but also to cover the
  gap between \inez{} and Z3 and surpass Z3.}

We conclude the evaluation by pointing out that there is significant
room for improvement in \inezdb{}. As is the case with almost every
first implementation of a new decision procedure, there is room for
improvement, \eg \inezdb{} can benefit from better preprocessing and
more sophisticated branching. \inezdb{} can also be improved by
adopting database techniques (as we outlined in
Section~\ref{sec:bct}), or by integrating a DBMS. Our promising
experimental results even without such optimizations constitute
sufficient evidence that ILP Modulo Data is a viable design for
data-enabled reasoning tools.

\comment{
  The data-specific portion of \inezdb{} is a straightforward
  implementation of the techniques described in Section~\ref{sec:bct} in
  less than 1500 lines of code. To name a few limitations,
  \begin{inparaenum}[(a)]
  \item we handle universal quantification by eager expansion;
  \item Our branching strategy is naive;
  \item as per Section~\ref{sec:bct}, our procedure only looks at one
    table and one witness row at a time;
  \item we always index the first column of each table, which is an
    arbitrary choice; and
  \item we apply no preprocessing.
  \end{inparaenum}
}

\section{Related Work}
\label{sec:related}


The Constraint Database framework~\cite{kanellakis90} provides a
database perspective on constraint solving. The framework encompasses
relations described by means of constraints, but not relations
comprised of concrete tuples.

``Table constraints''~\cite{tablecp,dstablecp}, as studied in
Constraint Programming, resemble our membership constraints. Such
tables are not meant as database tables. Our work differs in
significant ways, \eg our setup allows symbolic table
contents. Also, the algorithms presented for table constraints rely on
table contents from small domains (\ie not the reals or the
integers). This aligns with the overall emphasis of Constraint
Programming, but conflicts with our intended applications.

Veanes et al. describe the Qex technique and tool that uses Z3 to
generate tests for SQL queries~\cite{qex}. Qex essentially encodes the
relational operators via axioms, which are later instantiated via
E-matching~\cite{ematching}. E-matching is a generic scheme that is
not optimized in any way for database problems. Qex is geared towards
relatively small tables that suffice as test cases, while our target
applications involve bigger tables.

Other approaches tackle constraints arising in database applications
with off-the-shelf generic solvers (via eager reductions).  Notably,
Khalek et al. use Alloy~\cite{alloyqueries}, while Meliou and Suciou
use MILP~\cite{tiresias}. In neither of these approaches does the core
of the solver exploit the structure of database instances, \eg for
branching or propagation.

\balance

\section{Conclusions and Future Work}
\label{sec:conclusions}

We introduced the ILP Modulo Data framework for marrying data with
symbolic reasoning.  To that end, we introduced the decidable
logic~\D{}. \comment{ for reasoning about databases. \D{}~is powerful
  enough for many database analysis tasks. Deciding~\D{} is a daunting
  task, primarily because of quantifiers.} We identified a fragment of
\D{}~that can be solved efficiently by instantiating the \bct{}
architecture. We developed a solver for \D{}, and evaluated this
solver on a set of benchmarks that we made publicly available.

There are many interesting research directions to be explored in
future work, including:
\begin{inparaenum}[(a)]
\item the design and implementation of solvers that include an actual
  DBMS,
\item efficiently handling universal quantification over big tables,
  say by partitioning input tables and using parallelization,
\item extending our techniques to allow mixed integer, real
  arithmetic, and other first-order theories, and
\item solving interesting business and scientific applications using
  the ILP Modulo Data framework.
\end{inparaenum}

\bibliographystyle{plain}
\bibliography{paper}

\end{document}